\documentclass[a4paper,UKenglish,numberwithinsect,thm-restate]{lipics-v2021}
\pdfoutput=1 

\usepackage[textsize=small]{todonotes}

\usepackage{amsmath}
\usepackage{amssymb}
\usepackage{mathtools}
\usepackage{hyperref}
\usepackage{graphicx}
\usepackage{enumitem}
\usepackage{afterpage}
\usepackage{xcolor}
\usepackage{tikzscale}
\usepackage{empheq}

\newcommand{\Nn}{\mathcal{N}}
\newcommand{\Pp}{\mathcal{P}}

\newcommand{\Nat}{\mathbb{N}}
\newcommand{\Rpos}{\mathbb{R}_{\ge 0}}
\newcommand{\dom}{dom}

\newcommand{\xra}[1]{\xrightarrow{#1}}

\renewcommand{\int}[1]{\lfloor #1 \rfloor}
\renewcommand{\frac}[1]{\{ #1 \}}
\newcommand{\Sync}{\ensuremath{Sync}}
\newcommand{\incl}{\subseteq}
\newcommand{\preg}{\equiv^p_M}
\newcommand{\reg}{\equiv_M}
\newcommand{\region}[1]{[ #1]}
\newcommand{\prg}{\mathsf{ProdRegAut}}
\newcommand{\PSPACE}{\textsf{PSPACE}}
\newcommand{\ta}{\mathsf{TA}}

\renewcommand{\a}{\alpha}
\newcommand{\Oo}{\mathcal{O}}

\usepackage{tikz}
\usetikzlibrary{positioning}
\usetikzlibrary{math}

\title{A local-time semantics for negotiations}

\author{Madhavan Mukund}{Chennai Mathematical Institute, India\\ CNRS
  IRL 2000, ReLaX, Chennai,
  India}{sri@cmi.ac.in}{}{}

\author{Adwitee Roy}{Chennai Mathematical Institute,
  India}{adwitee@cmi.ac.in}{}{} 

\author{B. Srivathsan}{Chennai Mathematical Institute, India\\ CNRS
  IRL 2000, ReLaX, Chennai,
  India}{sri@cmi.ac.in}{}{}

\authorrunning{Mukund, Roy and Srivathsan}

\Copyright{Madhavan Mukund, Adwitee Roy and B. Srivathsan}


\category{}

 \relatedversion{}

 \nolinenumbers
 
\hideLIPIcs

\begin{document}
\tikzmath{\dx = 2;}

\maketitle

\begin{abstract}
  Negotiations, introduced by Esparza et al., are a model for
  concurrent systems where computations involving a set of agents
  are described in terms of their interactions.
  In many situations, it is natural to impose timing constraints
  between interactions --- for instance, to limit the time
  available to enter the PIN after inserting a card into an ATM.
  To model this, we introduce a real-time aspect to negotiations.
  In our model of \emph{local-timed negotiations}, agents have
  local reference times that evolve independently. Inspired by
  the model of networks of timed automata, each agent is equipped
  with a set of local clocks.  Similar to timed automata, the
  outcomes of a negotiation contain guards and resets over the
  local clocks.

  As a new feature, we allow some interactions to force the
  reference clocks of the participating agents to synchronize.
  This synchronization constraint allows us to model interesting
  scenarios. Surprisingly, it also gives unlimited computing
  power. We show that reachability is undecidable for local-timed
  negotiations with a mixture of synchronized and unsynchronized
  interactions. We study restrictions on the use of synchronized
  interactions that make the problem decidable.
\end{abstract}

\section{Introduction}

Computing systems often consist of multiple components that interact with each other to execute a task. For instance, ATMs, online banking platforms, and e-commerce retailers all maintain a coordinated conversation between customers at the front end and data servers at the back end. In many cases, these interactions need to meet timing constraints---for example, a one-time password (OTP) times out if it is not entered with a short window. Hence, when specifying such interactions, it becomes important to accurately describe the interplay between concurrency and timing.

In \cite{DBLP:journals/acta/DeselEH19,DBLP:conf/concur/EsparzaD13}, Esparza et al. introduced \textit{negotiations} as a model for describing computations involving a set of agents.  Conventional automata-theoretic models focus on states and transitions, and specify how local states of agents determine global behaviours.  In negotiations, the basic building blocks are the interactions between the agents. Individual interactions between a set of agents are called \emph{atomic negotiations}. After each atomic negotiation, the participating agents collectively agree on an outcome and move on to participate in other atomic negotiations. Apart from providing an attractive alternative perspective for modelling concurrent systems, negotiations also admit efficient analysis procedures. For some subclasses, interesting properties can be analyzed in polynomial-time~\cite{DBLP:journals/lmcs/EsparzaKMW18,DBLP:conf/lics/EsparzaMW17}.  

The basic negotiation model does not have any mechanism to incorporate timing constraints between interactions. In~\cite{DBLP:conf/fossacs/AkshayGHM20} a notion of timed negotiations has been proposed, where every outcome is associated with an interval representing the minimum and maximum amount of time required for the interaction to conclude. The work focuses on computing the minimum and maximum execution times for the overall negotiation to complete. This model cannot express constraints on the time between different atomic negotiations. For this, we introduce clocks, as defined in timed automata~\cite{DBLP:journals/tcs/AlurD94}. 

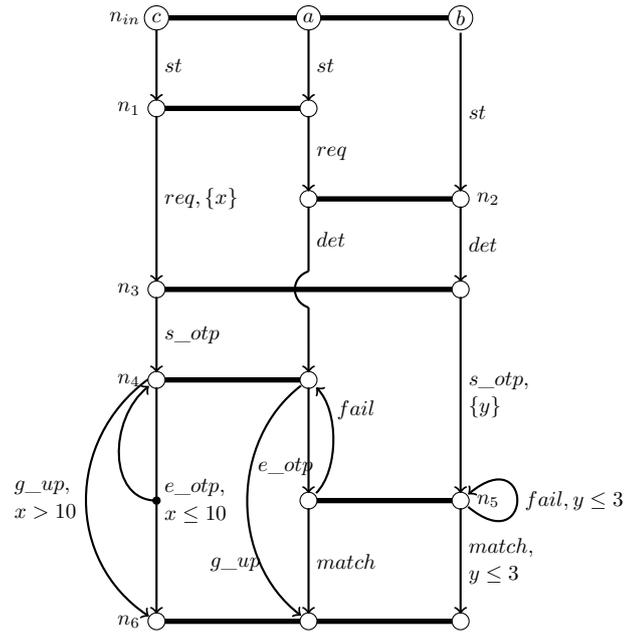
\begin{figure}[t]
\centering
\begin{tikzpicture}[scale=0.8, every node/.style={transform shape}]

\begin{scope}[line width=2pt]
\draw (0.,0) node [left] {$n_{in}\,\,$} -- (5,0);
\draw (0.,-1.5) node [left] {$n_{1}\,\,$} -- (2.5,-1.5);
\draw (2.5,-3) -- (5,-3) node [right] {$\,\,n_{2}$};
\draw (0.,-4.5) node [left] {$n_{3}\,\,$} -- (5,-4.5);
\draw (0.,-6) node [left] {$n_{4}\,\,$} -- (2.5,-6);
\draw (2.5,-8) -- (5,-8)  node [right] {\,\,$n_{5}$};
\draw (0.,-10) node [left] {$n_{6}\,\,$} -- (5,-10);
\end{scope}

\begin{scope}[fill=white]
\filldraw  (0,0)  circle (2mm) node (nic) {$c$};
\filldraw  (2.5,0)  circle (2mm) node (nia) {$a$};
\filldraw  (5,0)  circle (2mm) node (nib) {$b$};

\filldraw  (0,-1.5)  circle (1.4mm) node (n1c) {};
\filldraw  (2.5,-1.5)  circle (1.4mm) node (n1a) {};

\filldraw  (2.5,-3)  circle (1.4mm) node (n2a) {};
\filldraw  (5,-3)  circle (1.4mm) node (n2b) {};

\filldraw  (0,-4.5)  circle (1.4mm) node (n3c) {};
\filldraw  (5,-4.5)  circle (1.4mm) node (n3b) {};

\filldraw  (0,-6)  circle (1.4mm) node (n4c) {};
\filldraw  (2.5,-6)  circle (1.4mm) node (n4a) {};

\filldraw  (2.5,-8)  circle (1.4mm) node (n5a) {};
\filldraw  (5,-8)  circle (1.4mm) node (n5b) {};

\filldraw  (0,-10)  circle (1.4mm) node (n6c) {};
\filldraw  (2.5,-10)  circle (1.4mm) node (n6a) {};
\filldraw  (5,-10)  circle (1.4mm) node (n6b) {};

\filldraw  [fill=black] (0,-8)  circle (0.6mm) node (v1) {};

\end{scope}

\begin{scope}[thick,->]

\draw (nic) -- (n1c) node [right, midway] {$st$};
\draw (n1c) -- (n3c) node [right, midway] {$req, \{x\}$};
\draw (n3c) -- (n4c) node [right, midway] {$s\_otp$};
\draw (n4c) -- (n6c) node [right, text width = 1.2 cm, midway] {$e\_otp,$ $x\leq10$};
\draw (0,-8) .. controls  (-0.8,-8) and (-0.8,-6.65) .. (n4c) node [left, midway] {};

\draw (nia) -- (n1a) node [right, midway] {$st$};
\draw (n1a) -- (n2a) node [right, midway] {$req$};
\draw (2.5,-4.8) -- (n4a);
\draw (n4a) -- (n5a) node [left=-0.3cm, text width = 1 cm, near end] {$e\_otp$};
\draw (n5a) -- (n6a) node [right, text width = 1 cm, midway] {$match$};
\draw (n5a) .. controls (3,-7.5) and (3,-6.5) .. (n4a) node [right, near end] {$fail$};
\draw (n4a) .. controls  (1.2,-7) and (1.2,-9) .. (n6a) node [left=-0.3cm, text width = 1 cm, near end] {$g\_up$};

\draw (nib) -- (n2b) node [right, midway] {$st$};
\draw (n2b) -- (n3b) node [right, midway] {$det$};
\draw (n3b) -- (n5b) node [right, text width = 1 cm, midway] {$s\_otp,$ $\{y\}$};
\draw (n5b) -- (n6b) node [right, text width = 1 cm, midway] {$match,$ $y\leq3$};
\draw (n5b) .. controls  (6.2,-9) and (6.2,-7) .. (n5b) node [right, midway] {$fail, y\leq3$};

\draw (n4c) + (-0.14,0.01) .. controls  (-1.5,-7) and (-1.5,-9) .. (n6c) node [left=-0.2cm, text width = 1.25 cm, midway] {$g\_up,$ $x > 10$};

\end{scope}

\draw [thick] (n2a) -- (2.5,-4.2) node [right, midway] {$det$};
\draw [thick] (2.5,-4.2) .. controls (2.2,-4.3) and (2.2,-4.7) .. (2.5,-4.8);
\end{tikzpicture}
\caption{A local-timed negotiation modeling a transaction between a customer, an ATM and a bank}
\label{fig:ATM}
\end{figure}

\paragraph*{A motivating example.} We use the example in Figure~\ref{fig:ATM} to introduce our model informally.  Consider a time-constrained transaction in an ATM machine ($a$), where a customer ($c$) wants to reset her ATM PIN via an OTP received from her bank ($b$).  Here, $a$, $c$, and $b$ are the \emph{agents} in the system. Their direct interactions are represented by thick horizontal lines, called nodes. After each interaction, the participating agents decide on an outcome, represented by arrows. Initially, all agents are in the node $n_{in}$ node. They choose the outcome $st$ to start the transaction. Agents $a$ and $c$ go to node $n_1$ and $b$ goes to $n_2$.The customer gives her card details and requests for a PIN change in the ATM at  node $n_1$ by choosing the outcome $req$. At $n_2$, the ATM conveys this request to the bank and sends the details through the outcome $det$. The bank and the customer talk to each other at $n_3$, and $b$ sends an OTP to the customer through the outcome $s\_otp$. At $n_4$, the customer enters the OTP in the ATM. After entering the OTP, shown by the outcome $e\_otp$ the customer is ready to engage with the ATM either at $n_4$ or at $n_6$, represented by the non-deterministic arc leading to $n_4$ and $n_6$.
The ATM talks to the bank at $n_5$ to check the OTP. If it matches, the ATM goes to $n_6$, otherwise it goes back to $n_4$. 

In this example, we model two time constraints. The bank would like to ensure that at most $3$ time units elapse between the sending of the OTP and the final match. This is achieved by resetting a local clock $y$ of the bank and checking that $y \le 3$ at the outcome \emph{match}.  On the other hand, the customer wants at most $10$ time units to have elapsed between initiating the request and completing the transaction. This is achieved by resetting a local clock $x$ at the outcome $req$ and checking that $x \le 10$ at the outcome $e\_otp$. If more than $10$ units elapse, the customer gives up, fires the outcome $g\_up$.  It is natural to imagine that clocks $x$ and $y$ are local to the customer and to the bank and that they may evolve at different rates.  We formalize this behaviour in terms of our local-time semantics. However, as will see later, in certain interactions, it is useful to force the agents to synchronize their local times. Combining the concurrency present in negotiations with timing constraints and a mechanism to synchronize local times makes the model surprisingly powerful.

\paragraph*{Structure of the paper.} The paper is organized as
follows.  We begin by formalizing local-timed negotiations
(Section~\ref{sec:local-timed-negot}), with some illustrative
examples. We study the reachability problem for this model. We show
that when the negotiation has no interactions that synchronize local
times, or when all interactions force a synchronization of local
times, reachability is \PSPACE-complete
(Sections~\ref{sec:synchr-free-negot}
and~\ref{sec:always-synchr-negot}). In the general case when, there is
a mix of synchronized and unsynchronized interactions, reachability is
undecidable (Section~\ref{sec:undecidable}).

\paragraph*{Related work.}

A local-time semantics was proposed in the context of networks of timed automata in~\cite{DBLP:conf/concur/BengtssonJLY98}. Recently, the semantics has been applied to the zone-based verification of reachability~\cite{DBLP:conf/concur/GovindHSW19,DBLP:conf/lics/0001HSW22} and B\"uchi reachability properties~\cite{DBLP:conf/concur/HerbreteauSW22}. Local-time semantics in timed automata has been investigated as a basis for applying partial-order methods.
In our current work, we consider local-time semantics as the starting point and make synchronization an option that can be specified explicitly. This allows more independence between the agents and keeps the concurrency between actions of disjoint sets of agents. 

Models mixing concurrency and timing have been widely studied: time Petri nets~\cite{10.5555/907383}, timed-arc Petri nets~\cite{DBLP:conf/apn/Hanisch93}, time-constrained message sequence charts~\cite{DBLP:conf/concur/AkshayMK07} to name a few. Each model offers a different view of the computation. To our knowledge, a notion of a real-time has not yet been considered for negotiations. 

\section{Local-timed negotiations}
\label{sec:local-timed-negot}

For a finite set $S$ we write $\Pp(S)$ for the power set containing
all subsets of $S$. Let $\Rpos$ denote the set of non-negative reals,
and $\Nat$ the set of natural numbers. A \emph{clock} is a real-valued
variable whose values increase along with time and get updated during
transitions (exact semantics comes later). Let $X$ be a set of
clocks. A \emph{guard} over $X$ is a conjunction of clock constraints
of the form $x \bowtie c$ where $x \in X$,
$\bowtie \in \{ <, \le, =, >, \ge \}$ and $c \in \Nat$. We write
$\Phi(X)$ for the set of guards over $X$.

\begin{definition}[Local-timed
  negotations]\label{def:local-timed-negot}
  Let $P$ be a finite set of agents, $\Sigma$ a finite set of outcomes
  and $X$ a finite set of clocks. We assume that $X$ is partitioned as
  $\{ X_p \}_{p \in P}$ 
  with $X_p$ being the \emph{local clocks} for agent $p$. Further, we associate a
  special clock $t_p$ to each agent $p$, called its \emph{reference 
    clock}. This clock is neither used in a guard nor reset. A
  \emph{local-timed negotiation} $\Nn$ is given by a tuple
  $(N, \dom, \delta, Sync)$ where
  \begin{itemize}
  \item $N$ is a finite set of nodes (also called atomic
    negotiations); there is a special initial node $n_{in} \in N$,
  \item $\dom: N \to \Pp(P)$ maps each node to a non-empty subset of
    agents; we assume $\dom(n_{in}) = P$; for $p \in P$, we let
    $N_p := \{ n \in N \mid p \in \dom(n) \}$
  \item $\delta = \{\delta_p\}_{p \in P}$ is a tuple of transition
    relations, one for each agent, where
    $\delta_p: N_p \times \Sigma \to \Phi(X) \times \Pp(N_p) \times
    \Pp(X_p)$ maps each node-outcome pair $(n, a)$ to a guard
    $g \in \Phi(X)$, a set of nodes $M \incl N_p$ that $p$ becomes
    ready to engage in after this outcome, and a set $Y \incl X_p$ of
    clocks that get reset; we will call node-outcome pairs $(n, a)$ as
    \emph{locations},
  \item $\Sync \incl N$ is a subset of \emph{synchronizing nodes}.
  \end{itemize}
\end{definition}

Figure~\ref{fig:ATM} gives an example of a local-timed negotiation
over agents $P = \{c, a, b\}$. The set of nodes is given by 
$N = \{n_{in}, n_1, \dots, n_6\}$.  The domain $\dom(n)$ of a node $n$
is represented by the ``circles'' in each node: for instance,
$\dom(n_1) = \{c, a\}$ and $\dom(n_2) = \{a, b\}$. Agent $c$ has a
local clock $x$, and agent $b$ has a local clock $y$. As an example of
a transition for agent $c$, we have $\delta_c(n_4, e\_otp) = (x \le
10, \{n_4, n_6\}, \{\})$. There is a guard $x \le 10$, the agent
is ready to engage in $n_4$ and $n_6$ after the outcome, and no clock
is reset. In this example, $\Sync$ is empty. 

\paragraph*{Semantics.} The semantics of a negotiation is described using \emph{markings} and
\emph{valuations}.  A marking $C$ is a function assigning each agent
$p$ to a subset of $N_p$. It gives the set of nodes that each agent is
ready to engage in. A valuation $v: X \cup T \to \Rpos$ maps every
clock (including reference clocks) to a non-negative real such that
$v(x) \le v(t_p)$ for all $x \in X_p$, and all agents $p \in P$. The
interpretation is that clocks in $X_p$ move at the same pace as $t_p$,
the local reference clock. Since $t_p$ is never reset it gives the
local time at agent $p$. For a constraint $x \bowtie c$ we say
$v \models x \bowtie c$ if $v(x) \bowtie c$. We say $v$ satisfies
guard $g \in \Phi(X)$, written as $v \models g$, if $v$ satisfies
every atomic constraint appearing in $g$.

A \emph{local-delay} $\Delta \in \Rpos^{|P|}$ is a vector of
non-negative reals, giving a time elapse for each agent. Each agent
can have a different time elapse.  Given a valuation $v$ and a
local-delay $\Delta$, we write $v + \Delta$ for the valuation obtained
as: for each agent $p \in P$, we have
$(v + \Delta)(y) = v(y) + \Delta(p)$ for every $y \in t_p \cup
X_p$. Notice that all clocks within a process move at the same rate as
its reference clock. However, the reference clocks of different agents
can move at different speeds.
For a set of clocks $Y \incl X$, we denote by $v[Y]$ the valuation
satisfying $v[Y](y) = 0$ if $y \in Y$ and $v[Y](y) = v(y)$ otherwise. 

A configuration is a pair $(C, v)$ consisting of a marking $C$ and a
valuation $v$. The \emph{initial configuration} $(C_0, v_0)$ contains
a marking $C_0$ which maps every agent to $n_{in}$ and valuation $v_0$
maps all clocks to $0$. We write $(C, v) \xra{\Delta} (C, v+\Delta)$
for the local-delay transition $\Delta$ at configuration $(C, v)$.
For the negotiation in Figure~\ref{fig:ATM}, an example of a
configuration is $(\bar{C}, \bar{v})$ with $\bar{C}(c) = \{ n_1 \}$,
$\bar{C}(a) = \{n_1\}$ and $\bar{C}(b) = \{n_2\}$ and $\bar{v}(t_c) =
\bar{v}(x) = 2$, $\bar{v}(t_a) = 1$ and $\bar{v}(t_b) = \bar{v}(y) =
3$. Suppose $\Delta = (1, 0, 2)$, then $\bar{v} + \Delta$ maps $t_c$
and $x$ to $3$, and $t_b$ and $y$ to $5$ whereas $t_a$ remains $1$.

 A location $\ell = (n, a)$ can be executed at a configuration $(C, v)$
leading to a configuration $(C', v')$, written as
$(C, v) \xra{\ell} (C', v')$, provided there is an entry
$\delta_p(n, a) = (g_p, M_p, Y_p)$ for all $p \in \dom(n)$ such that:
\begin{itemize}
\item \emph{current marking enables the negotiation:} $n \in C(p)$ for
  all $p \in \dom(n)$,
\item \emph{synchronization condition is met:} if $n \in Sync$, then
  $v(t_p) = v(t_q)$ for all $p, q \in \dom(n)$,
\item \emph{guard is satisfied}: $v \models g_p$ for all
  $p \in \dom(n)$,
\item \emph{target marking is correct}: $C'(p) = M_p$ for all
  $p \in \dom(n)$, $C'(p) = C(p)$ for $p \notin \dom(n)$,
\item \emph{resets are performed}: $v'(y) = 0$ for
  $y \in \bigcup_{p \in \dom(n)} Y_p$
\end{itemize}

For an example, consider Figure~\ref{fig:ATM} again and a
configuration $(C^1, v^1)$ with $C^1:= (\{n_4, n_6\}, \{n_4\},
\{n_6\})$ and $v^1: \langle t_c = 10, x = 5, t_a = 20, t_b = 5, y = 1
\rangle$. Location $(n_4, e\_otp)$ is enabled leading to a
configuration $(C^2, v^2)$ where $C^2 := ( \{n_4, n_6\}, \{n_6\},
\{n_6\})$ and $v^2 = v^1$, as there are no resets in this location.  

We call $(C, v) \xra{\Delta} (C, v + \Delta) \xra{\ell} (C', v')$ a
\emph{small step} and write this as
$(C, v) \xra{\Delta, \ell} (C', v')$ for conciseness. A \emph{run} is
a sequence of small steps starting from the initial configuration. We
say that a location $\ell = (n,a)$ is \emph{reachable} if there is a
run containing a small step that executes $\ell$.

\emph{Reachability problem.} We are interested in the following
question: given a location $\ell = (n,a)$, is it reachable?

\subsection{Some examples}

In the example of Figure~\ref{fig:ATM}, we have seen how
local-clocks can be used to constrain interactions. We will now see
some examples that show some interesting mechanics of synchronized
interactions. The negotiation in the left of Figure~\ref{fig:k-vendor} has three agents $p, q,
v$. The outcome at node $n_1$ results in a non-deterministic choice
for agent $q$: the agent may either decide to talk with $p$ at $n_3$
or with $v$ at $n_2$. Suppose at $n_3$, agent $p$ wants to make sure
that $q$ has arrived at $n_3$ after talking to $v$. We can imagine
that $v$ is a vendor, and $p$ wants to ensure that $q$ has indeed
met the vendor between their meetings at $n_1$ and $n_3$. To do this,
we make use of timing and synchronization constraints as
follows.

\begin{figure}[t]
\centering
\begin{tikzpicture}[scale=0.65, every node/.style={transform shape}]

\begin{scope}[shift={(-6,0)}]
\begin{scope}[line width=2pt]
\draw (0.,0) node [left] {$n_{in}\,\,$} -- (4.5,0);
\draw [densely dotted, opacity=0.4] (0,-1.5) node [black, left, opacity=1] {$n_1\,\,$} -- (1.5,-1.5);
\draw (3,-4.5)  -- (4.5,-4.5) node [right] {\,\,$n_2$};
\draw [densely dotted, opacity=0.4] (0,-6) node [black, left, opacity=1] {$n_3\,\,$} -- (1.5,-6);
\draw (0,-7.5) node [left] {$n_{4}\,\,$} -- (1.5,-7.5);
\end{scope}

\filldraw [fill=white]  (0,0)  circle (2mm) node 		(nip) {$p$};
\filldraw [fill=white]  (1.5,0)  circle (2mm) node 		(niq) {$q$};
\filldraw [fill=white]  (4.5,0)  circle (2mm) node 		(niv) {$v$};

\filldraw [fill=white]  (0,-1.5)  circle (1.4mm) node 	(n1p) {};
\filldraw [fill=white]  (1.5,-1.5)  circle (1.4mm) node (n1q) {};

\filldraw [fill=black]  (1.5,-3) circle (.6mm) node		(virt) {};

\filldraw [fill=white]  (3,-4.5)  circle (1.4mm) node 	(n21) {};
\filldraw [fill=white]  (4.5,-4.5)  circle (1.4mm) node (n2v) {};

\filldraw [fill=white]  (0,-6)  circle (1.4mm) node 	(n3p) {};
\filldraw [fill=white]  (1.5,-6)  circle (1.4mm) node 	(n3q) {};

\filldraw [fill=white]  (0,-7.5)  circle (1.4mm) node 	(nfp) {};
\filldraw [fill=white]  (1.5,-7.5)  circle (1.4mm) node (nfq) {};

\begin{scope}[thick,->]

\draw (nip) -- (n1p);
\draw (n1p) -- (n3p) node [midway, left, text width = 1cm] {$x=2$};
\draw (n3p) -- (nfp) node [midway, left] {};

\draw (niq) -- (n1q);
\draw (n1q) -- (n3q) node [very near start, right] {$y=0$};
\draw (virt) + (-.05,.05) -- (n21);
\draw (n21) -- (n3q) node [midway, above] {$\{y\}$};
\draw (n3q) -- (nfq) node [midway, right, text width = 1cm] {$y=0$};

\draw (niv) -- (n2v);

\end{scope}
\end{scope}

\begin{scope}[line width=2pt]
\draw (0.,0) node [left] {$n_{in}\,\,$} -- (7.9,0);
\draw (8.8,0) node [left] {} -- (10.5,0);
\draw [densely dotted, opacity=0.4] (0,-1.5) node [black, left, opacity=1] {$n_1\,\,$} -- (1.5,-1.5);
\draw (3,-4.5)  -- (4.5,-4.5) node [right] {\,\,$m_1$};
\draw (5.5,-4.5)  -- (7,-4.5) node [right] {\,\,$m_2$};
\draw (9,-4.5)  -- (10.5,-4.5) node [right] {\,\,$m_k$};
\draw [densely dotted, opacity=0.4] (0,-6) node [black, left, opacity=1] {$n_3\,\,$} -- (1.5,-6);
\draw (0,-7.5) node [left] {$n_{4}\,\,$} -- (1.5,-7.5);
\end{scope}

\filldraw [fill=white]  (0,0)  circle (2mm) node 		(nip) {$p$};
\filldraw [fill=white]  (1.5,0)  circle (2mm) node 		(niq) {$q$};
\filldraw [fill=white]  (4.5,0)  circle (2mm) node 		(niv) {$v_1$};
\filldraw [fill=white]  (7,0)  circle (2mm) node 		(niv2) {$v_2$};
\filldraw [fill=white]  (10.5,0)  circle (2mm) node 		(nivk) {$v_k$};

\filldraw [fill=white]  (0,-1.5)  circle (1.4mm) node 	(n1p) {};
\filldraw [fill=white]  (1.5,-1.5)  circle (1.4mm) node (n1q) {};

\filldraw [fill=black]  (1.5,-2.5) circle (.6mm) node		(virt) {};

\filldraw [fill=white]  (3,-4.5)  circle (1.4mm) node 	(n21) {};
\filldraw [fill=white]  (4.5,-4.5)  circle (1.4mm) node (n2v) {};

\filldraw [fill=white]  (0,-6)  circle (1.4mm) node 	(n3p) {};
\filldraw [fill=white]  (1.5,-6)  circle (1.4mm) node 	(n3q) {};

\filldraw [fill=white]  (0,-7.5)  circle (1.4mm) node 	(nfp) {};
\filldraw [fill=white]  (1.5,-7.5)  circle (1.4mm) node (nfq) {};

\filldraw [fill=white]  (5.5,-4.5)  circle (1.4mm) node 	(v2q) {};
\filldraw [fill=white]  (7,-4.5)  circle (1.4mm) node (v2v) {};

\filldraw [white, fill=white]  (7.5,-4.8)  circle (1.4mm) node (virt2) {};
\filldraw [white, fill=white]  (8.5,-4.8)  circle (1.4mm) node (virt3) {};
\filldraw [white, fill=white]  (8.5,-5.5)  circle (1.4mm) node (virt8) {};

\filldraw [fill=black]  (3.7,-5.22)  circle (.6mm) node (virt4) {};
\filldraw [fill=black]  (5.94,-5.22)  circle (.6mm) node (virt5) {};

\filldraw [white, fill=white]  (7.5,-5.1)  circle (1.4mm) node (virt6) {};
\filldraw [white, fill=white]  (7.5,-5.5)  circle (1.4mm) node (virt7) {};

\filldraw [fill=white]  (9,-4.5)  circle (1.4mm) node 	(vkq) {};
\filldraw [fill=white]  (10.5,-4.5)  circle (1.4mm) node (vkv) {};

\begin{scope}[thick,->]

\draw (nip) -- (n1p);
\draw (n1p) -- (n3p) node [very near start, left] {$a$};
\draw (n3p) -- (nfp) node [near start, right, text width = 1cm] {$b$};

\draw (niq) -- (n1q);
\draw (n1q) -- (n3q) node [very near start, right] {$a$};
\draw (virt) + (-.05,.05) -- (n21);
\draw (n3q) -- (nfq) node [near start, right, text width = 1cm] {$b$};

\draw (niv) -- (n2v);

\draw (niv2) -- (v2v);
\draw (nivk) -- (vkv);

\draw (1.5,-2.5) -- (v2q);
\draw (1.5,-2.5) -- (vkq);

\draw (n21) -- (n3q) node [near start, left = 0cm] {$a_1$};

\end{scope}

\draw [line width=2pt, loosely dotted] (7.9,0)  -- (8.8,0) {};
\draw [very thick, dotted] (7.9,-4.5)  -- (8.4,-4.5) {};
\draw [very thick, dotted] (8.7,-5)  -- (8.7,-5.3) {};
\draw [very thick, dotted] (5.3,-5.9)  -- (5.3,-6.3) {};

\draw [thick,->] (n21) .. controls (3.5,-5.5) and(4.5,-5.5) .. (v2q) node [very near start, right = 0.1cm] {$b_1$};
\draw [thick,->] (v2q) .. controls (5.8,-5.5) and(6.5,-5.5) .. (virt2) node [very near start, right = 0cm] {$b_2$};
\draw [thick,->] (virt3) .. controls (8.7,-4.9) and (8.9,-4.8) .. (vkq);
\draw [thick,->] (virt8) .. controls (8.9,-5.4) and (9.1,-4.8) .. (9.1,-4.6);
\draw [thick,->] (v2q) .. controls (5,-6) and (3,-6) .. (n3q) node [midway, above left= -0.1cm] {$a_2$};
\draw [thick,->] (3.7,-5.22) .. controls (5,-6) and (6.5,-6) .. (virt6);
\draw [thick,->] (5.94,-5.22) .. controls (6,-6) and (6.5,-6) .. (virt7);
\draw [thick,->] (9.1,-4.6) .. controls (9,-7) and (5,-6.5) .. (n3q) node [midway, above left = -0.1cm] {$a_k$};;

\end{tikzpicture}
\caption{In the figure on the right,
          the $a$ transition has guards $x = m$ for $p$ and $y = 0$
          for $q$. Each $b_i$ transition has a guard $y=1$ and a reset
          of $y$ to ensure exactly $1$ unit of time is spent in the
          nodes $m_i$ before outcomes $b_i$. The outcomes $a_i$ have a
          reset of $y$. The $b$ transition from
          $n_3$ has a guard $y=0$. }   
	\label{fig:k-vendor}
\end{figure}

We first make $n_1$ and $n_3$ as synchronization nodes, that is, they
are part of $\Sync$ for this negotiation. In the picture we represent
it as the coloured nodes. Suppose $x$ is a clock of agent $p$ and $y$ is a clock of
agent $q$. At
$n_1$, the outcome checks for the guard $x = 2$ and $y = 0$. When this
outcome is fired, the local-clock $t_p$ is ahead of $t_q$ by $2$
units. Now, we make $n_3$ a synchronizing node and add a guard $y = 0$
in the outcome of $n_3$. If $q$ comes to $n_3$ directly after talking
to $p$ at $n_1$, then we have $t_q = y = 0$, but $t_p = 2$. No time
can elapse at $q$ since there is a $y = 0$ guard. But then, the
synchronization condition cannot be satisfied. This forces $q$ to meet
$v$ at $n_2$, spend sufficient time ($2$ units, in this case), reset
the clock $y$ and then interact with $p$ at $n_3$.

This example can be extended in the case where there are multiple
vendors $v_1, \dots v_k$ and $p$ wants $q$ to have met at least $m$
vendors out of them before resynchronizing, as shown in Figure
\ref{fig:k-vendor} on the right. We also assume that once
$q$ interacts with $v_i$, she cannot interact with any vendor $v_j$
with $j \le i$. If each interaction of $q$ with a vendor $v_i$ takes
$1$ time unit, we can force the clock of $p$ to be at least $m$ at
node $n_1$. Therefore at $n_3$, the only way for $q$ to ensure
synchronization with $p$, and have $y = 0$ is by finding $m$ other
interactions where she can spend time.

In Figure~\ref{fig:abc} we present an example which has been used in
different contexts dealing with a partial-order semantics for timed
automata~\cite{DBLP:journals/tcs/LugiezNZ05} or the local-time semantics for networks of timed
automata~\cite{DBLP:conf/lics/0001HSW22}.  
Outcomes $a$ and $b$ are local to agents $p$ and $q$ whereas
$c$ is the result of a negotiation. 
We make node $n_3$ a synchronizing
node and ask for the guard $x = 1$ and $y = 1$ at $c$. If $t_p = n, x
= 1$ for agent $p$, we want $t_q = n, y = 1$ at agent $q$. This
constraint forces the same number of $a$s and $b$s to have happened
locally before $p$ and $q$ interact at $n_3$. There is no ordering
relation between the $a$s and $b$s, for instance we cannot say 
that the
second $a$ happens after the first $b$. Therefore the untimed language
is simply the language of all words with the same number of $a$s and
$b$s before the $c$. This shows that the
untimed language of the outcome sequences need not even be regular,
unlike the case of timed automata. 

\begin{figure}[t]
\centering
\begin{tikzpicture}[scale=0.9, every node/.style={transform shape}]

\begin{scope}[line width=2pt]
\draw (-.25,0) -- (0.25,0);
\draw (2.25,0) -- (2.75,0);
\draw [opacity=0.4, densely dotted] (-.25,-2) -- (2.75,-2);
\end{scope}

\filldraw [fill=white]  (0,0)  circle (1.5mm) node (pa) {$p$};
\filldraw [fill=white]  (2.5,0)  circle (1.5mm) node (qb) {$q$};
\filldraw [fill=white]  (0,-2)  circle (1.5mm) node (pc) {};
\filldraw [fill=white]  (2.5,-2)  circle (1.5mm) node (qc) {};

\node at (0.45,0) {$n_1$};
\node at (2.05,0) {$n_2$};
\node at (2.95,-2) {$n_3$};
\node at (-1.3,0.6) [text width=1cm] {$a,$};
\node at (3.9,0.6) [text width=1cm] {$b,$};
\node at (-1.3,0) [text width=1cm] {$x=1,$ $\{x\}$};
\node at (3.9,0) [text width=1cm] {$y=1,$ $\{y\}$};

\begin{scope}[thick,->]
\draw (pa) + (0,-.14) -- (pc) node [midway, left] {$a,$} node [midway, right, text width=1cm] {$x=1,$ $\{x\}$};
\draw (qb) + (0,-.14) -- (qc) node [midway, left] {$b,$} node [midway, right, text width = 1cm] {$y=1,$ $\{y\}$};
\draw (pc) -- (0,-3.5) node [midway, left] {$c,$} node [midway, right, text width=1cm] {$x=1$};
\draw (qc) -- (2.5,-3.5) node [midway, left] {$c,$} node [midway, right, text width=1cm] {$y=1$};
\draw (-0.1,-0.1) .. controls (-1,-1) and (-1,1) .. (-0.1,0.1);
\draw (2.6,-0.1) .. controls (3.5,-1) and (3.5,1) .. (2.6,0.1);
\end{scope}

\end{tikzpicture}
\caption{A local-timed negotiation depicting that the untimed
          language of the outcome sequences need not be regular. The
          synchronizing node $n_3$ forces the number of $a$s and number of $b$s to be equal.}

	\label{fig:abc}
\end{figure}
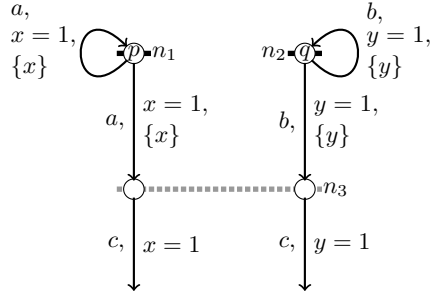

\section{Synchronization-free negotiations}
\label{sec:synchr-free-negot}

Our goal is to study the location reachability problem in local-timed
negotiations. Before studying the general case, we look at some
restricted versions. The first restriction we look at is a
synchronization-free fragment. Fix a negotiation
$\Nn = (N, \dom, \delta, \Sync)$ for this section.

We say that $\Nn$ is \emph{synchronization-free} if
$\Sync = \emptyset$. In such a negotiation, the agents require to
elapse sufficient time only to satisfy their local guards (and not to
meet any synchronization criteria). For instance, in the example on
the left of Figure~\ref{fig:k-vendor}, if $n_3$ is not a synchronizing
node, then $q$ can come directly to $n_3$ after node $n_1$, elapse no
time at all, and engage in the only outcome from $n_3$. In the
negotiation of Figure~\ref{fig:abc}, if the outcome $c$ is not
synchronized, the untimed language is the set of $wc$ where
$w \in (a + b)^*$ contains at least one $a$ and one $b$.

Although the time elapse needed for an agent is to satisfy her own
guards, she may need to collaborate with a partner to decide on what
amount to elapse. This is because of the combination of guards across
the different agents at a location. This is apparent in the
negotiation of Figure~\ref{fig:nonsync}. For $c$ to be feasible, both
the agents have to reach node $n_4$ via $n_2$, and not via
$n_3$. Therefore, one can view the time elapse at $n_1$ as a
collective decision between $p$ and $q$, which impacts their future
paths.

\begin{figure}[t]
\centering 
\begin{tikzpicture}[scale=0.9, every node/.style={transform shape}]

\begin{scope}[line width=2pt]
\draw (0,0) node [left] {$n_1\,\,$} -- (1.5,0);
\draw (-1.5,-1.5) node [left] {$n_2\,\,$} -- (0,-1.5);
\draw (1.5,-1.5) node [left] {$n_3\,\,$} -- (3,-1.5);
\draw (0,-3) node [left] {$n_4\,\,$} -- (1.5,-3);
\end{scope}

\filldraw [fill=white]  (0,0)  circle (2.2mm) node 		(n1p) {$p$};
\filldraw [fill=white]  (1.5,0)  circle (2.2mm) node 		(n1q) {$q$};

\filldraw [fill=white]  (-1.5,-1.5)  circle (1.5mm) node 	(n2p) {};
\filldraw [fill=white]  (0,-1.5)  circle (1.5mm) node 	(n2q) {};

\filldraw [fill=white]  (1.5,-1.5)  circle (1.5mm) node 	(n3p) {};
\filldraw [fill=white]  (3,-1.5)  circle (1.5mm) node 	(n3q) {};

\filldraw [fill=white]  (0,-3)  circle (1.5mm) node 		(n4p) {};
\filldraw [fill=white]  (1.5,-3)  circle (1.5mm) node 	(n4q) {};

\filldraw [white, fill=white]  (0,-4.5)  circle (1.5mm) node 	(vp) {};
\filldraw [white, fill=white]  (1.5,-4.5)  circle (1.5mm) node 	(vq) {};

\begin{scope}[thick,->]

\draw (n1p) + (-0.15,-0.15) -- (n2p);
\draw (n1q) + (-0.15,-0.15) -- (n2q);

\draw (n1p) + (0.15,-0.15) -- (n3p);
\draw (n1q) + (0.15,-0.15) -- (n3q);

\draw (n2p) -- (n4p) node [midway, left] {$x \geq 2$};
\draw (n2q) -- (n4q) node [near start, left] {$y \leq 1$};

\draw (n3p) -- (n4p) node [near start, right] {$x \leq 1$};
\draw (n3q) -- (n4q) node [midway, right] {$y \geq 2$};

\draw (n4p) -- (vp) node [midway, right] {$c$} node [midway, left] {$x = 2$};
\draw (n4q) -- (vq) node [midway, left] {$c$} node [midway, right] {$y = 1$};
\end{scope}

\end{tikzpicture}
\caption{Example of a synchronization-free local-timed
    negotiation. If $n_3$ is fired then the guard $y = 1$ on the
    transition from $n_4$ can not be satisfied.}
\label{fig:nonsync}
\end{figure}
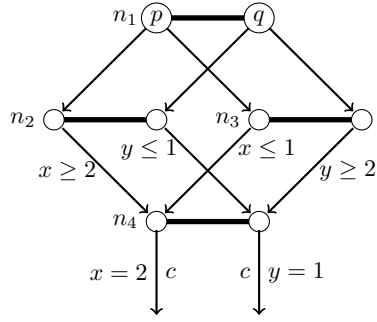

The goal of this section is to show that reachability is
\PSPACE-complete for synchronization-free negotiations. Here is an
overview of our proof. Firstly, as there are no synchronization
constraints, we observe that the reference clocks are not useful at
all in this fragment. We will then quotient the space of valuations by
applying the classical region equivalence between clocks of each
agent. This generates a finite automaton that accepts all the untimed
location sequences that are feasible in the negotiation.

\begin{definition}[$\equiv^p_M$ and $\reg$ equivalences]
  Let $p$ be an agent, and let $M$ be the biggest constant appearing
  in the negotiation. We say $v \preg v'$ if for all $x, y \in X_p$:
  \begin{itemize}
  \item either $\int{v(x)} = \int{v'(x)}$ or both
    $\int{v(x)}, \int{v'(x)} > M$,
  \item if $v(x) \le M$, then $\frac{v(x)} = 0$ iff $\frac{v'(x)} = 0$
    
  \item if $v(x) \le M$ and $v(y) \le M$, we have
    $\frac{v(x)} \le \frac{v(y)}$ iff $\frac{v'(x)} \le \frac{v'(y)}$.
  \end{itemize}
  We define $v \reg v'$ if $v \preg v'$ for all agents $p \in P$. We
  denote by $\region{v}$ the equivalence class of a valuation $v$ with
  respect to the $\reg$ equivalence. We will call $\reg$ as the
  \emph{product-region} equivalence, and the equivalence classes of
  $\reg$ as product-regions.
\end{definition}

Notice that reference clocks do not appear at all in the above
definition. We next state there are finitely many product-regions.

\begin{lemma}\label{lem:size-of-product-regions}
The $\reg$ equivalence is of finite index: the number of
product-regions is bounded by $\Oo(|X|! \cdot 2^{|X|} \cdot
(2M + 1)^{|X|})$.
\end{lemma}
\begin{proof}
  Let us call an equivalence class wrt $\preg$ as a $p$-region. This
  equivalence ignores values of clocks other than the local-clocks of
  $p$. Therefore, any two valuations that differ only in values of
  clocks outside  $X_p$ will be equivalent. For clocks in $X_p$ the
  equivalence is simply the classical region equivalence of timed
  automata. Hence, the number of $p$-regions equals the number of
  regions (with the same maximum constant), which is bounded by
  $K_p := |X_p|! \cdot 2^{|X_p|} \cdot (2M + 1)^{|X_p|}$ (Lemma 4.5,
  \cite{DBLP:journals/tcs/AlurD94}). Now, by definition, $v \reg v'$
  if both $v$ and $v'$ are in 
  the same $p$-region for every agent $p$. Hence each product-region
  can be seen as a tuple consisting of a $p$-region for each $p$. So,
  the number of product-regions is bounded by $\prod_{p \in P}
  K_p$. This can be shown to be bounded by $|X|! \cdot 2^{|X|} \cdot
(2M + 1)^{|X|}$, using the fact that $|X_1|! \cdot |X_2|! \le (X_1 +
X_2)!$, and $c^{|X_1|} \cdot c^{|X_2|} = c^{|X_1| + |X_2|}$ for any
constant $c$. \qed  
\end{proof}

Here are some properties of the product-region equivalence, that
follow from the region equivalence.

\begin{restatable}{lemma}{localDelay}
\label{lemma:region-local-delay}
  Let $v, v'$ be valuations such that $v \reg v'$. Then, for all
  local-delays $\Delta$, there exists a local-delay $\Delta'$ such
  that $v + \Delta \reg v' + \Delta'$.
\end{restatable}
\begin{proof}
  The delay $\Delta$ can be seen as a sequence of local-delays
  $(\Delta(p_1), 0, \dots, 0)$, $(0, \Delta(p_2), 0, \dots, 0)$
  $\dots$ $(0, \dots, 0, \Delta(p_k))$ where only one agent makes a
  delay in each step. Therefore, it is sufficient to show the lemma when $\Delta$
  has a non-zero delay only for one process. Assume
  $\Delta(p) = \delta \ge 0$, and $\Delta(p') = 0$ for all
  $p' \neq p$. Notice that in $v + \Delta$ the values of clocks of
  agents different from $p$ do not change. If we restrict to the local
  clocks of agent $p$ (without the reference clock), we can use the
  classical region equivalence to get a $\delta'$ and a local-delay
  satisfying $\Delta(p) = \delta'$, $\Delta(p') = 0$ for all
  $p' \neq p$ such that $v + \Delta \reg v' + \Delta'$.  \qed
\end{proof}

The next lemma follows by definition.
\begin{lemma}\label{lem:region-guard-reset}
  Let $v, v'$ be valuations such that $v \reg v'$. Let $g$ be a guard
  with constants at most $M$. Then: (1) $v \models g$ iff
  $v' \models g$, and (2) for all subsets of local clocks $Y$, we have
  $v[Y] \reg v'[Y]$.
\end{lemma}

For a valuation $v$ and a product-region $r$, we write $v \in r$ to
mean that $r$ equals $[v]$. We will now build a finite automaton using
the product-regions.

\begin{definition}[Product-region automaton]\label{def:prod-reg-aut}
  States of this NFA are of the form $(C, r)$ where $C$ is a marking
  and $r$ is a product-region. There is a transition
  $(C, r) \xra{(n, a)} (C', r')$ if for some valuation $v \in r$, and
  for some local-delay $\Delta$, we have
  $(C, v) \xra{\Delta, (n, a)} (C', v')$ such that $v' \in r'$. The
  initial state is the initial marking $C_0$ and the region $r_0$
  containing the valuation that maps all clocks to $0$.

  We denote the product-region automaton as $\prg(\Nn)$.
\end{definition}

\begin{restatable}{lemma}{negToReg}
\label{lem:neg-to-reg}
  For every run
  $(C_0, v_0) \xra{\Delta_0, \ell_0} (C_1, v_1) \xra{\Delta_1, \ell_1}
  \cdots \xra{\Delta_{m-1}, \ell_{m-1}} (C_m, v_m)$ in the local-timed
  negotiation $\Nn$, there is a run
  $(C_0, [v_0]) \xra{\ell_0} (C_1, [v_1]) \xra{\ell_1} \cdots
  \xra{\ell_{m-1}} (C_m, [v_m])$ in $\prg(\Nn)$.
\end{restatable}
\begin{proof}
  Follows from Definition~\ref{def:prod-reg-aut}.
\end{proof}

\begin{restatable}{lemma}{regToNeg}
    \label{lem:reg-to-neg}
  For every run
  $(C_0, r_0) \xra{\ell_0} (C_1, r_1) \xra{\ell_1} \cdots
  \xra{\ell_{m-1}} (C_m, r_m)$ in the product-region automaton
  $\prg(\Nn)$, there is a run
  $(C_0, v_0) \xra{\Delta_0, \ell_0} (C_1, v_1)$
  $\xra{\Delta_1, \ell_1} \cdots \xra{\Delta_{m-1}, \ell_{m-1}} (C_m,
  v_m)$ in $\Nn$ such that $v_i \in r_i$ for each $0 \le i \le m$.
\end{restatable}
\begin{proof}
  Assume we have constructed a run
  $(C_0, v_0) \xra{\Delta_0, \ell_0} \cdots \xra{\Delta_{i-1},
    \ell_{i-1}} (C_i, v_i)$ with $v_j \in r_j$ for all
  $0 \le j \le i \le m$. By definition of the transitions of
  $\prg(\Nn)$, there is some $u_i \in r_i$ and some local-delay
  $\theta_i$ s.t.
  $(C_i, u_i) \xra{\theta_i, \ell_i} (C_{i+1}, u_{i+1})$. By
  Lemmas~\ref{lemma:region-local-delay} and
  \ref{lem:region-guard-reset}, there exists a $\Delta_i$ satisfying 
  $(C_i, v_i) \xra{\Delta_i, \ell_i} (C_{i+1}, v_{i+1})$ with
  $v_{i+1} \in r_{i+1}$. This shows that we can extend the run one
  step at a time to get a run with as required by the lemma.
  \qed
\end{proof}

\begin{theorem}
  Reachability is \PSPACE-complete for synchronization-free
  local-timed negotiations.
\end{theorem}
\begin{proof}
  A location $\ell = (n,a)$ is reachable in $\Nn$ iff it is reachable
  in $\prg(\Nn)$, thanks to Lemmas~\ref{lem:neg-to-reg} and
  \ref{lem:reg-to-neg}. Let $K$ be the size of the negotiation counted
  as the sum of the number of nodes, outcomes, clocks and the sum of
  the binary encoding of each constant appearing in the guards. 

  We can non-deterministically guess a path from
  the initial node to an edge containing $\ell$. Each node $(C, r)$
  requires polynomial space: $C$ can be represented as a vector of bit
  strings, one for each agent. Each bit string gives the set of nodes
  where the agent can engage in, and hence has a length equal to the
  number of nodes. The region uses constants of size at most
  $M$. The size of $\prg(\Nn)$ is
  a product of the number of markings, and the number of regions for
  each agent. Both the number of markings and the number of
  product-regions is $2^{\Oo(K)}$
  (Lemma~\ref{lem:size-of-product-regions}). These two arguments sum
  up to give the \PSPACE\ upper   bound.

  \PSPACE-hardness follows from the hardness of timed automata
  reachability. Each timed automaton can be seen as a negotiation over
  a single agent. Therefore, reachability in timed automata can be
  reduced to reachability in a local-timed negotiation (when there is
  a single agent, both the notions of local-time or global-time
  coincide). \qed 
\end{proof}
 
\section{Always-synchronizing negotiations}
\label{sec:always-synchr-negot}

We will now look at the fragment where every interaction forces a
synchronization.  We say that a local-timed negotiation is
\emph{always-synchronizing} if every node is a synchronization node,
that is $\Sync = N$. The negotiation in Figure~\ref{fig:abc} can be
seen as an always-synchronizing negotiation (in nodes that are local
to one agent, the synchronization condition is vacuously true). We
first remark that a region based argument is not immediate in this
fragment. In order to satisfy the synchronization constraint, we check
conditions of the form $t_p = t_q$. Therefore, we cannot decouple the
time elapse of $p$ and $q$ completely, as in the previous
section. Instead, we need to keep track of the difference $t_p - t_q$
in the equivalence. But then, there is no bound $M'$ that allows us to
club together all valuations with $t_p - t_q > M'$. This is because,
$t_q$ can perform local delays to catch up with $p$, and in
particular, from $t_p - t_q > M'$ we may get to a situation where
$t_p - t_q \le M'$. This kind of a mechanics does not happen in
classical timed automata, where once a clock is beyond $M$ it always
stays beyond $M$ until the next reset. In the previous section, we
avoided this problem since we did not need to keep track of the
reference clocks.

We will make use of a different argument altogether, which is used in
the proof of equivalence between the local-time and global-time
semantics for networks of timed
automata~\cite{DBLP:conf/concur/BengtssonJLY98,DBLP:conf/concur/GovindHSW19}. In
always-synchronizing negotiations, every location $(n, a)$ is executed
at a unique timestamp given by the reference clock value of the
participating processes. For example, the sequence $aabbc$ in the
negotiation of Figure~\ref{fig:abc} can be associated with the
timestamp $12123$: the first $a$ occurs at $t_p = 1$, the second $a$
at $t_p = 2$, the first $b$ at $t_q = 1$ and so on.  The main
observation is that whenever there is a $t_i t_{i+1}$ in this sequence
with $t_{i+1} < t_i$, we can reorder the actions corresponding to
them, and still get a valid run. For example, the run
$(a, 1) (a, 2) (b, 1) (b, 2) (c, 3)$ described above can be reordered
as $(a, 1) (b, 1) (a, 2) (b, 2) (c, 3)$ which is still a feasible run
of the negotiation.

We will first show that every run of an always-synchronizing sequence
can be reordered to a ``monotonic'' run. Next, we describe a timed
automaton that accepts all monotonic runs of the negotiation. This
gives a procedure for reachability, as reachability in the negotiation
reduces to checking whether there is a run of a timed automaton that
fires an edge.

\begin{definition}
  Let $\Nn$ be an always-synchronizing negotiation. Consider a run
  $\rho:= (C_0, v_0) \xra{\Delta_0, \ell_0} (C_1, v_1) \xra{\Delta_1,
    \ell_1} \cdots \xra{\Delta_{m-1}, \ell_{m-1}} (C_m, v_m)$, where
  $\ell_i = (n_i, a_i)$ for every $i$.

  We associate timestamps $\theta^\rho_i := v_i(t_p) + \Delta_i(p)$
  where $p$ is some agent participating in the negotiation $n_i$.  The
  run $\rho$ is monotonic if $\theta^\rho_i \le \theta^\rho_j$ for
  every $i \le j$.
\end{definition}

Fix an always-synchronizing negotiation $\Nn$ for the rest of this
section.

\begin{restatable}{lemma}{monotonic}
\label{lem:monotonic}
  For every location $(n,a)$ that is reachable, there is a monotonic
  run containing a small step that executes $\ell$.
\end{restatable}
\begin{proof}
  Let
  $\rho := (C_0, v_0) \xra{\Delta_0, \ell_0} (C_1, v_1) \xra{\Delta_1,
    \ell_1} \cdots \xra{\Delta_{m-1}, \ell_{m-1}} (C_m, v_m)$ be a run
  such that $\ell_{m-1} = (n, a)$. Let us write $\dom(\ell_j)$ for the
  agents that participate in the negotiation corresponding to location
  $\ell_j$.

  Without loss of generality, we can
  assume that for every $0 \le j \le m-1$, we have $\Delta_j(p) = 0$
  if $p \notin \ell_j$. Indeed, if $\Delta_j(p)$ is
  non-zero, we can move this time elapse to the immediate next
  position in the run where $p$ participates, and if $p$ does not
  participate in any position to the right, we can simply change the
  delay to $0$ and still preserve the feasibility of the sequence.

  Consider a segment $\sigma:= (C_j, v_j) \xra{\Delta_j, \ell_j} (C_{j+1},
  v_{j+1}) \xra{\Delta_{j+1}, \ell_{j+1}} (C_{j+2}, v_{j+2})$ such
  that $\dom(\ell_j) \cap \dom(\ell_{j+1})$. We claim that the two
  outcomes can be commuted to give a run where the end points are the
  same: that is, a run of the form $\sigma':= (C_j, v_j) 
  \xra{\Delta_{j+1}, \ell_{j+1}} (C'_{j+1}, v'_{j+1}) \xra{\Delta_j,
    \ell_j} (C_{j+1}, v_{j+2})$. Here is an argument. Suppose $p
  \notin \dom(\ell_j) \cup \dom(\ell_{j+1})$. For all $x \in X_p$ we have
  $v_j(x) = v_{j+1}(x) = v_{j+2}(x)$. The same is true in
  $\sigma'$. Suppose $p \in \dom(\ell_j)$. Then by assumption, $p
  \notin \dom(\ell_{j+1})$. Therefore $\Delta_{j+1}(p) =
  0$. This suffices to prove the claim. 
 
  Coming back to the run $\rho$. Suppose there is an index $j$ such
  that $\theta^\rho_{j+1} < \theta^\rho_{j}$. Then there are no common
  agents in $\ell_j$ and $\ell_{j+1}$. We can apply the commutation
  argument and swap the two small steps. We can keep doing this until
  there is no index which violates monotonicity. 
  Therefore, for every run $\rho$ of the negotatiation, we are able to
  associate a monotonic run $\rho'$ as described above containing all
  the locations. The location $\ell_{m-1} = (n, a)$ therefore appears
  somewhere in the run.  \qed
\end{proof}

\begin{definition}
  For an always-synchronizing negotiation $\Nn$ we define a timed
  automaton $\ta(\Nn)$ as follows.  States are the set of all markings
  possible in $\Nn$. There is a transition $C \xra{g, (n,a) , Y} C'$
  on guard $g$, action $(n,a)$ and reset $Y$ if (1) $n$ is enabled in
  $C$, (2) there are transitions $\delta_p(n, a) = (g_p,M_p, Y_p)$ and
  $g$ is the conjunction of all $g_p$, and $Y$ is the union of all
  $Y_p$, (3) $C'(p) = M_p$ for $p \in \dom(n)$ and $C'(p) = C(p)$
  otherwise.
\end{definition}

\begin{restatable}{lemma}{alwaysTA}
\label{lem:ta-for-always-sync}
  Let $\Nn$ be an always-synchronizing negotiation. A location $(n,a)$
  is reachable in $\Nn$ iff there is a run in the timed automaton
  $\ta(\Nn)$ that executes an edge labeled with $(n, a)$.
\end{restatable}
\begin{proof}
  Suppose $(n,a)$ is reachable in $\Nn$. By Lemma~\ref{lem:monotonic},
  there is a monotonic run
  $\rho:= (C_0, v_0) \xra{\Delta_0, \ell_0} \cdots \xra{\Delta_{m-1},
    \ell_{m-1}} (C_m, v_m)$ such that $\ell_{m-1} = (n, a)$. The run
  $\rho$ in itself is not a run of $\ta(\Nn)$ since in a timed
  automaton all clocks increase by the same amount, whereas here, the
  time delays are still local and asynchronous. However, we can 
  massage this run to make it a run of the timed automaton.

  We have $v_0$ as the initial valuation, which maps every clock to
  $0$. Let $\delta_0 = \theta^\rho_0$ and
  $\delta_i = \theta^\rho_i - \theta^\rho_{i-1}$ for $i \ge 1$. Due to
  the monotonicity assumption, we have $\delta_i \ge 0$. We let all
  agents elapse time $\delta_i$ at the $i^{th}$ step. We
  claim that the run
  $(C_0, u_0) \xra{\delta_0, \ell_0} (C_1, u_1) \xra{\delta_1, \ell_1}
  \cdots \xra{\delta_{m-1}, \ell_{m-1}} (C_m, u_m)$ with $u_0 = v_0$
  is a run of the timed automaton. This follows from the observation
  that $(u_i + \delta_i)(x) = (v_i + \Delta_i(p))(x)$ for all clocks
  $x \in X_p$ and all agents $p \in \dom(\ell_i)$. This proves the
  left-to-right direction.

  Suppose
  $(C_0, u_0) \xra{\delta_0, \ell_0} (C_1, u_1) \xra{\delta_1, \ell_1}
  \cdots \xra{\delta_{m-1}, \ell_{m-1}} (C_m, u_m)$ is a run of
  $\ta(\Nn)$ with $\ell_{m-1} = (n, a)$. Taking a local-delay
  $\Delta_i$ that maps every agent to $\delta_i$ gives us the same
  sequence as a run in $\Nn$, thereby proving the right-to-left
  direction. \qed

\end{proof}

\begin{theorem}
  Reachability is \PSPACE-complete for always-synchronizing
  negotiations.
\end{theorem}
\begin{proof}
  From Lemma~\ref{lem:ta-for-always-sync}, it is enough to check
  reachability of a certain edge in $\ta(\Nn)$. The idea is to
  non-deterministically guess a path in the region automaton of
  $\ta(\Nn)$.

  Let $K$ be the size of the negotiation $\Nn$ which
  includes the number of nodes, outcomes, clocks and the sum of the
  binary encodings of the constants present. 
  The number of states of $\ta(\Nn)$ is $2^{\Oo(K)}$. The set of
  clocks is the same as that of $\Nn$. Therefore, the number of
  regions for $\ta(\Nn)$ is still $2^{\Oo(K)}$. The product of states
  and regions remains to be $2^{\Oo(K)}$. Therefore the guesses path
  is of size bounded by $2^{\Oo(K)}$.  
  Moreover, each node of the region automaton can be represented in
  polynomial space. This gives the \PSPACE\ upper bound. 

  Lower bound follows once again from the hardness of timed
  automata, which is simply a negotiation with a single
  agent. Synchronization is vacuously true at every node. \qed
\end{proof}

\section{Reachability is undecidable for local-timed negotiations}
\label{sec:undecidable}

When we allow both synchronized and unsynchronized nodes, we are
unable to use either of the techniques of the previous two
sections. In fact, reachability
turns out to be undecidable. Since local-times are independent of each
other, it is possible to have an unbounded drift between the reference
clocks of two agents. This helps store counter values as differences
between the local times. The main challenge is the check for
zero. This is where we require a combination of synchronized and
unsynchronized interactions. We will now show to simulate a counter
machine using a local-timed negotiation.

\begin{theorem}\label{thm:undecidability}
  Reachability is undecidable for local-timed negotiations.
\end{theorem}

The rest of the section is devoted to proving
Theorem~\ref{thm:undecidability}. We will encode the halting problem
of a $2$-counter machine as the reachability problem for a local-timed
negotiation.

\paragraph*{Counter machines.} A $2$-counter machine $M$ is a program that manipulates two counters, $C_1$ and $C_2$,
each of which can hold a non-negative number. The machine is given as
a sequence of labelled instructions $\ell:I$, where $I$ is one of the following for some $i \in \{1, 2\}$: 
\begin{description}
\item[increment] $C_i + +$, which increments the value of the
  counter $C_i$ and goes to the next instruction
  $\ell+1$.
\item[decrement] if $C_i>0$ then $C_i- -$, which decrements $C_i$ and continues with the next instruction
  $\ell+1$. If the value of $C_i$ is $0$, then the program is blocked. 
\item[jump-on-zero] if $C_i= = 0$ goto $\ell'$, which transfers control to the instruction labelled $\ell'$ if counter $C_i$ is $0$ for
  $i\in\{1,2\}$. If $C_i>0$, it continues to the instruction $\ell+1$.
\end{description}
The counter machine is said to halt if it reaches the final
instruction. A configuration of $M$ is a triple $(\ell, c_1, c_2)$ representing the current instruction $\ell$ that needs to be executed and the current values $c_1, c_2 \ge 0$ of the counters $C_1, C_2$ respectively. The transitions $(\ell, c_1, c_2) \xra{} (\ell', c'_1, c'_2)$ follow naturally from the description above.

\paragraph*{Overview of the reduction.} The negotiation $\Nn_M$ that we
construct will have $6$ agents $p_1, q_1, r_1, p_2, q_2, r_2$. Agents
$p_1, q_1, r_1$ simulate counter $C_1$, and the rest simulate
$C_2$. Let $i \in \{1, 2\}$. The local clocks of $p_i, q_i, r_i$ are
respectively $\{x_{p_i}\}$, $\{x_{q_i}, x'_{q_i}\}$ and $\{x_{r_i}\}$. Additionally, we have the reference clocks $t_\a$ for each agent $\a$. For every instruction $\ell$, we will have a node
$n_\ell$ in which all the six agents participate.
A configuration $(C, v)$ of $\Nn_M$ is said to encode configuration
$(\ell, c_1, c_2)$ of $M$ if:
\begin{itemize}
\item $C(\a) = \{n_\ell\}$ for every agent $\a$,
\item $v(x) = 0$ for all local clocks (and reference clocks can take any value),
\item $v(t_{r_1}) \le v(t_{q_1}) \le v(t_{p_1})$ and $v(t_{r_2}) \le v(t_{q_2}) \le v(t_{p_2})$, and 
\item $v(t_{p_1} - t_{q_1}) = c_1$ and $v(t_{p_2} - t_{q_2}) = c_2$,
\end{itemize}
The initial configuration of $\Nn_m$ has every agent in $n_{\ell_0}$, where $\ell_0$ is the first instruction in $M$ and every clock (including reference clocks) to be $0$.

We will have a gadget in $\Nn_M$ corresponding to each instruction in the counter machine. Let $(C, v), (C',v')$ be configurations that encode $(\ell, c_1, c_2)$ and $(\ell', c'_1, c'_2)$ respectively. A run $(C, v) \xra{} \cdots \xra{} (C', v')$ such that none of the intermediate configurations encodes any counter machine configuration will be called a \emph{big step}. We denote a big step as $(C, v) \Rightarrow (C', v')$. Our gadgets will ensure the following two properties.
\begin{itemize}
\item Let $(\ell, c_1, c_2) \xra{} (\ell', c'_1, c'_2)$ in $M$. Then from every configuration $(C, v)$ that encodes $(\ell, c_1, c_2)$, there is a big step $(C, v) \Rightarrow (C', v')$ to some configuration $(C', v')$ that encodes $(\ell', c'_1, c'_2)$.
\item Let $(C, v), (C', v')$ be arbitrary configurations of $\Nn_M$ that encode $(\ell, c_1, c_2)$ and $(\ell', c'_1, c'_2)$ respectively. If $(C, v) \Rightarrow (C', v')$ is a big step, then $(\ell, c_1, c_2) \xra{} (\ell', c'_1, c'_2)$ in $M$.
\end{itemize}
The first property ensures that for every path $(\ell^0, c^0_1, c^0_2)
\xra{} (\ell^1, c^1_1, c^1_2) \xra{} \cdots$, there is a sequence of
big steps $(C_0, v_0) \Rightarrow (C_1, v_1) \Rightarrow \cdots$ such
that  $(C_i, v_i)$ encodes $(\ell^i, c^i_1, c^i_2)$. The second
property ensures the reverse: from a sequence of big steps, we get
corresponding run in counter machine. We will now describe each gadget and show that the two properties are satisfied. 

\begin{figure}[t]
\begin{center}
\begin{tikzpicture}[scale=0.8, every node/.style={transform shape}]

\begin{scope}[line width=2pt]
\draw (0,0) -- (7.5,0);
\draw (0,-2.1) -- (7.5,-2.1);
\end{scope}

\filldraw [fill=white]  (0,0)	  circle (2.2mm) node (nip1) [label={left:$n_{\ell}$}] {$p_1$};
\filldraw [fill=white]  (1.5,0)	  circle (2.2mm) node (niq1) {$q_1$};
\filldraw [fill=white]  (3,0)	  circle (2.2mm) node (nir1) {$r_1$};
\filldraw [fill=white]  (4.5,0)	  circle (2.2mm) node (nip2) {$p_2$};
\filldraw [fill=white]  (6,0)	  circle (2.2mm) node (niq2) {$q_2$};
\filldraw [fill=white]  (7.5,0)	  circle (2.2mm) node (nir2) {$r_2$};
\filldraw [fill=white]  (0,-2.1)	  circle (1.4mm) node (n2p1) [label={left:$n_{\ell + 1}$}] {};
\filldraw [fill=white]  (1.5,-2.1)	  circle (1.4mm) node (n2q1) {};
\filldraw [fill=white]  (3,-2.1)	  circle (1.4mm) node (n2r1) {};
\filldraw [fill=white]  (4.5,-2.1)	  circle (1.4mm) node (n2p2) {};
\filldraw [fill=white]  (6,-2.1)	  circle (1.4mm) node (n2q2) {};
\filldraw [fill=white]  (7.5,-2.1)	  circle (1.4mm) node (n2r2) {};

\begin{scope}[thick,->]
  \draw (nip1) -- node [very near start, right] {$a$} (n2p1)	 node [midway,right, text width=1.25cm] {$x_{p_1}=1$ $\{x_{p_1}\}$};
  \draw (niq1) -- node [very near start, right] {$a$} (n2q1)	 node [midway,right, text width=1.25cm] {$x_{q_1}=0$};
  \draw (nir1) -- node [very near start, right] {$a$} (n2r1)	 node [midway,right, text width=1.25cm] {$x_{r_1}=0$};
  \draw (nip2) -- node [very near start, right] {$a$} (n2p2)	 node [midway,right, text width=1.25cm] {$x_{p_2}=0$};
  \draw (niq2) -- node [very near start, right] {$a$} (n2q2)	 node [midway,right, text width=1.25cm] {$x_{q_2}=0$};
  \draw (nir2) -- node [very near start, right] {$a$} (n2r2)	 node [midway,right, text width=1.25cm] {$x_{r_2}=0$};
\end{scope}

\draw [dashed, thick] (-1.2,0.5) -- (9,0.5) -- (9,-1.8) -- (-1.2,-1.8) -- (-1.2,0.5);

\node at (4,.75) {Increment Gadget for $c_1$};
\end{tikzpicture}
\caption{Gadget for implementing increment instruction on $c_1$}
\label{fig:increment}
\end{center}
\end{figure}
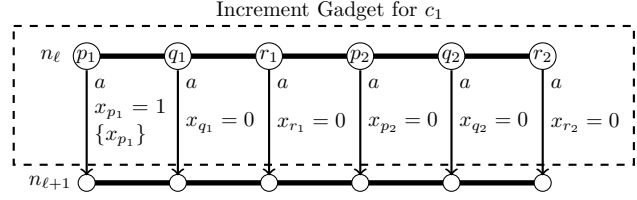

\paragraph*{Increment.} Assume an instruction $\ell: C_1 ++$ with $\ell$ not being the final instruction. The case with $C_2++$ is symmetric. Every configuration $(\ell, c_1, c_2)$ on executing this instruction goes to $(\ell + 1, c_1 + 1, c_2)$. Figure~\ref{fig:increment} shows the gadget for the increment instruction. All agents other than $p_1$ cannot elapse time due to the guard checking local clocks to $0$. Agent $p_1$ elapses
exactly one time unit, after which clock $x_{p_1}$ is reset.
It is easy to see that the configuration $(C', v')$ that results from $(C, v)$
encodes $(\ell', c_1 + 1, c_2)$. The big step $(C, v) \Rightarrow (C',v')$ is in fact a single transition. Both the properties are easily seen to be satisfied. 

\begin{figure}[t]
\centering
\begin{tikzpicture}[scale=0.8, every node/.style={transform shape}]

\begin{scope}[line width=2pt]
\draw (0,0) -- (2.2*\dx+0,0);
\draw [loosely dotted] (2.2*\dx+0,0) -- (2.5*\dx+0,0);
\draw [densely dotted, opacity=0.4] (0,-1*\dx) -- (2*\dx+0,-1*\dx);
\draw (0,-2*\dx) -- (\dx+.3,-2*\dx);
\draw [densely dotted, opacity=0.4] (\dx-0,-3*\dx) -- (2*\dx+0,-3*\dx);
\draw (0,-4*\dx) -- (2*\dx,-4*\dx);
\draw (0,-5*\dx) -- (\dx,-5*\dx);
\draw (-2.6*\dx,-4.8*\dx) -- (-2.1*\dx,-4.8*\dx);
\draw (-2.6*\dx,-5.3*\dx) -- (-2.1*\dx,-5.3*\dx);
\draw (-2.1*\dx,-4.8*\dx)  -- (-1.9*\dx,-4.8*\dx);
\draw [loosely dotted] (-1.9*\dx,-4.8*\dx)  -- (-1.6*\dx,-4.8*\dx);
\draw (-2.1*\dx,-5.3*\dx)  -- (-1.9*\dx,-5.3*\dx);
\draw [loosely dotted] (-1.9*\dx,-5.3*\dx)  -- (-1.6*\dx,-5.3*\dx);
\draw (2.35*\dx,-5*\dx) -- (2.5*\dx,-5*\dx);
\draw [loosely dotted] (2.2*\dx,-5*\dx) -- (2.35*\dx,-5*\dx);
\end{scope}

\filldraw [fill=white]  (0,0)    circle (2.2mm) node (nip1) [label={left:$n_{\ell}$}] {$p_1$};
\filldraw [fill=white]  (\dx,0)  circle (2.2mm) node (niq1) {$q_1$};
\filldraw [fill=white]  (2*\dx,0) circle (2.2mm) node (nir1) {$r_1$};

\filldraw [fill=white] (0,-\dx)			circle	 (1.4mm)	 node	 (z1p1)	{};
\filldraw [fill=white] (2*\dx,-\dx)		circle	 (1.4mm)	 node	 (z1r1) [label={right:$\,\,\textcolor{black}{n_1}$}]	{};

\filldraw [fill=white] (0,-2*\dx)		circle	 (1.4mm)	 node	 (z2p1)	{};
\filldraw [fill=white] (\dx,-2*\dx)		circle	 (1.4mm)	 node	 (z2q1) [label={right:$\,\,n_2$}]	{};

\filldraw [fill=white] (\dx,-3*\dx)		circle	 (1.4mm)	 node	 (z3q1)	{};
\filldraw [fill=white] (2*\dx,-3*\dx)	circle	 (1.4mm)	 node	 (z3r1) [label={right:$\,\,\textcolor{black}{n_3}$}]	{};

\filldraw [fill=white] (0,-4*\dx)		circle	 (1.4mm)	 node	 (z4p1)	{};
\filldraw [fill=white] (\dx,-4*\dx)		circle	 (1.4mm)	 node	 (z4q1)	{};
\filldraw [fill=white] (2*\dx,-4*\dx)	circle	 (1.4mm)	 node	 (z4r1) [label={right:$n_4$}]	{};

\filldraw [fill=white] (0,-5*\dx)		circle	 (1.4mm)	 node	 (dp1)	{};
\filldraw [fill=white] (\dx,-5*\dx)		circle	 (1.4mm)	 node	 (dq1)  [label={right:$n_6$}]	{};

\filldraw [fill=white] (-2.6*\dx,-4.8*\dx) circle (1.4mm) node (z5p1) [label={left:$n_5$}]	{};
\filldraw [fill=white] (-2.1*\dx,-4.8*\dx)   circle (1.4mm) node (z5q1)	{};
\filldraw [fill=white] (-2.6*\dx,-5.3*\dx)   circle (1.4mm) node (nfp1) [label={left:$n_{fin}$}]	{};
\filldraw [fill=white] (-2.1*\dx,-5.3*\dx)     circle (1.4mm) node (nfq1)	{};

\filldraw [fill=white] (2.5*\dx,-5*\dx)	circle (1.4mm)  node (n5r1) [label={right:$n_5$}] {};

\begin{scope}[thick,->]
  \draw (nip1) -- (z1p1) node [midway, right, text width=1.25cm] {$st,$ $x_{p_1}=0$};
  \draw (z1p1) -- (z2p1) node [midway, right, text width=1.25cm] {$x_{p_1}=0$};
  \draw (z2p1) -- (z4p1) node [midway, left, text width = 1.25cm] {$a,$ $x_{p_1} = 0$};
  \draw (z4p1) -- (dp1)  node [midway, right, text width=1.25cm] {$x_{p_1}=0$};
  \draw (dp1)  -- (0,-6*\dx) node [midway, right, text width=1.25cm] {$x_{p_1}=0,$ $\{x_{p_1}\}$};
  \draw (z4p1) -- (-2.58*\dx,-4.75*\dx);
  \draw (z5p1) -- node [midway, right] {} (nfp1) node [midway, left] {$x_{p_1}\!\!<\!0$};
  \draw (z2p1) .. controls (-1.1,-2*\dx-1.1) and (-1.1,-2*\dx+1.1) .. (z2p1) node [midway, left=-0.4cm, text width=2cm] {$b,$ $x_{p_1}=1,$ $\{x_{p_1}\}$};
  
  \draw (\dx,-\dx-.2) -- (z2q1) node [midway, right] {};
  \draw (z2q1) -- node [midway, left] {$x_{q_1} = 0,$} (z3q1) node [midway, right] {$a$};
  \draw (z3q1) -- (z4q1) node [midway, right, text width=1.25cm] {$c,$\\ $x_{q_1}\!\!=\!0$};
  \draw (z4q1) --  (dq1) node [midway, right, text width=1.25cm] {$x_{q_1} = 0,$ $x_{q_1}'\!\!>\!0,$ $\{x_{q_1}'\}$};
  \draw (dq1)  -- (\dx,-6*\dx) node [midway, right, text width=1.25cm] {$x_{q_1}=1,$ $\{x_{q_1}\}$};
  \draw (z4q1) -- (z5q1) node [midway, below, sloped] {$x_{q_1}'=0$};
  \draw (z5q1) -- (nfq1);
  \draw (z2q1) .. controls (\dx+1.1,-2*\dx-1.1) and (\dx+1.1,-2*\dx+1.1) .. node [midway, right=-0.11cm, text width=1.25cm] {$b,$\\ $x_{q_1}\!\!=\!1,$ $\{x_{q_1}\}$} (z2q1);
  
  \draw (nir1) -- (z1r1) node [midway, right, text width=1cm] {$st,$ $x_{r_1}\!\!=\!0$};
  \draw (z1r1) -- (z3r1) node [midway, right] {$\{x_{r_1}\}$};
  \draw (z3r1) -- (z4r1) node [midway, right, text width=1.25cm] {$c,$\\ $x_{r_1}\!\!=\!0$};
  \draw (z4r1) -- (2*\dx,-6*\dx);
  \draw (z4r1) -- (n5r1);
\end{scope}

\draw [thick] (niq1) -- (\dx,-\dx+.2) node [midway, right, text width=1.25cm] {$st,$ $x_{q_1} = 0,$ $x_{q_1}'=0$};
\draw [thick] (\dx,-\dx+.2) .. controls (\dx+.1,-\dx+.2-.1) and(\dx+.1,-\dx-.2+.1) .. (\dx,-\dx-.2);

\draw [dashed, thick] (-1.5*\dx+.2,1) -- (3*\dx,1) -- (3*\dx,-4.8*\dx) -- (-1.5*\dx+.2,-4.8*\dx) -- (-1.5*\dx+.2,1) node [midway, sloped, above] {Zero-check gadget for $c_1$};

\draw [dashed, thick] (-.2*\dx,-4.85*\dx) -- (1.8*\dx,-4.85*\dx) -- (1.8*\dx,-5.85*\dx) -- (-.2*\dx,-5.85*\dx) -- (-.2*\dx,-4.85*\dx) node [midway, left, text width = 1.6cm] {Decrement gadget\\ for $c_1$};

\draw [dashed, thick] (-3.3*\dx,-4.6*\dx) -- (-1.5*\dx,-4.6*\dx) -- (-1.5*\dx,-5.5*\dx) -- (-3.3*\dx,-5.5*\dx) -- (-3.3*\dx,-4.6*\dx);

\draw [dashed, thick] (-.2*\dx,-6.05*\dx) -- (2.5*\dx,-6.05*\dx) -- (2.5*\dx,-6.75*\dx) -- (-.2*\dx,-6.75*\dx) -- (-.2*\dx,-6.05*\dx);

\node at (1.15*\dx,-6.4*\dx) {Gadget for implementing $\ell+1$};
\node at (-2.6*\dx,-4.4*\dx) {Block gadget for $c_1$};

\end{tikzpicture}
\caption{Gadget for implementing decrement instruction on $c_1$ }
\label{fig:jump-decrement}
\end{figure}
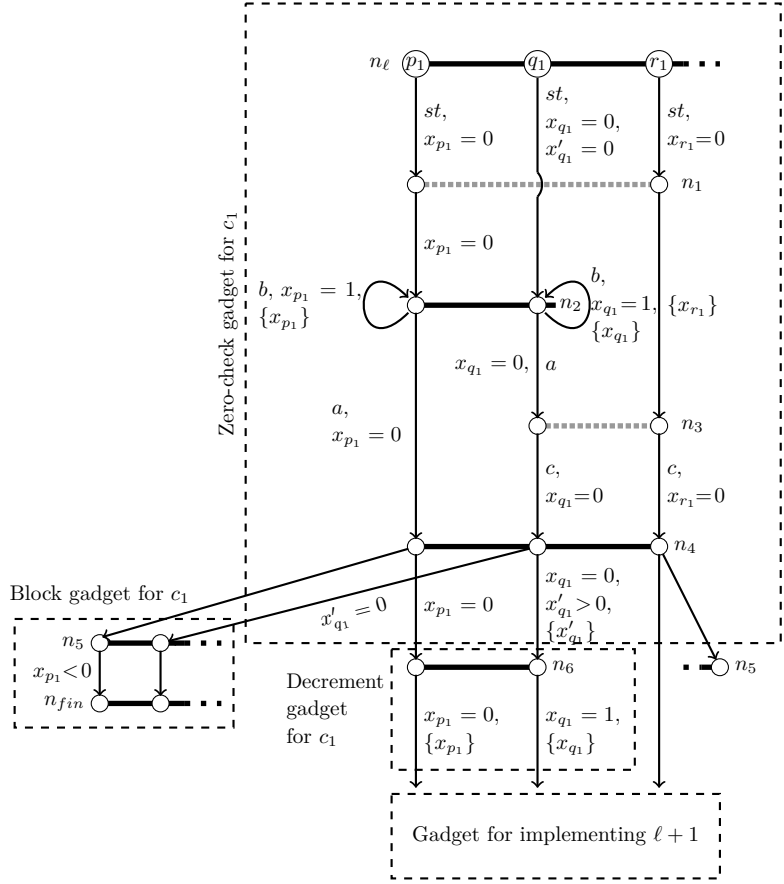

\paragraph*{Decrement.} Consider a decrement instruction: if $C_1 > 0$
then $C_1 --$. We have
$(\ell, c_1, c_2) \xra{} (\ell+1, c_1 - 1, c_2)$ if $c_1 > 0$. The
first task is to check if $c_1 > 0$. Recall that in the negotiation the difference $t_{p_1} - t_{q_1}$ gives the value of $c_1$. Our idea is to let $q_1$ elapse time to
synchronize with $p_1$ and check if this time elapse needed for
synchronization is strictly above $0$ or not. However, in this
process, we lose the actual value of the counter. In order to maintain the same difference between $t_{p_1}$ and $t_{q_1}$, we make use of the auxiliary process $r_1$.

Consider a configuration $(C, v)$ that encodes $(\ell, c_1, c_2)$. By our definition, $v(t_{r_1}) \le v(t_{p_1})$ and $v(t_{p_1} - t_{q_1}) = c_1$. 
\begin{itemize}
 
\item We first let $r_1$ synchronize with $p_1$, while $p_1$
  elapses no time.
\item Next, we keep moving both $p_1$ and $q_1$ by $1$ unit each until
  $q_1$ synchronizes with $r_1$. In this entire process $r_1$ is not
  allowed to elapse time.
\end{itemize}
By the end of this, we will get a valuation $v'$ with the same difference $v'(t_{p_1} - t_{q_1}) = c_1$
since both the agents were moved by the same amount. Moreover, we can
use an additional clock to check whether in the process of $q_1$
synchronizing with $r_1$, a non-zero time had elapsed in $q_1$.

The gadget is depicted in Figure~\ref{fig:jump-decrement}. For simplicity, we do not index the intermediate nodes $n_1, n_2, n_3, n_4$ by $\ell$. The computation proceeds in three phases. Below, we show the run of $\Nn_M$ along this gadget, restricted to the agents $p_1, q_1, r_1$. The other three agents simply move to the next possible instruction (this is not shown in the figure for simplicity). 

\paragraph*{Phase 1.} Synchronize $r_1$ with $p_1$ maintaining no time
elapse in $p_1$ as follows: $((n_\ell, n_\ell, n_\ell), v) \xra{}
((n_1, n_2, n_1), v_1) \xra{} ((n_2, n_2, n_3), v_3)$. After the last
action, agents $p_1$ and $r_1$ are synchronized, that is,
$v_3(t_{p_1}) = v_3(t_{r_1})$. Moreover, $v_3(t_{p_1}) = v(t_{p_1})$.

\paragraph*{Phase 2.} Move $p_1$ and $q_1$ repeatedly by one unit each:  $((n_2, n_2, n_3), v_3) \xra{b} ((n_2, n_2, n_3), v_3^1) \xra{b} \cdots \xra{b} ((n_2, n_2, n_3), v_3^k)$. By the end of $k$ iterations of $b$, we have $v^k_3(t_{p_1}) = v(t_{p_1}) + k$ and $v^k_3(t_{q_1}) = v(t_{q_1}) + k$.

\paragraph*{Phase 3.} Check if the reference clocks of $q_1$ and $r_1$
are equal: $((n_2, n_2, n_3), v_3^k)$
$\xra{a} ((n_4, n_3, n_3), v_4) \xra{c} ((n_4, n_4, n_4), v_5)$. The
outcome $c$ at $n_4$ can be fired only if the reference clocks of
$q_1$ and $r_1$ are equal: that is, $v_4(t_{q_1}) = v_4(t_{r_1})$. Due
to the guard checking for $0$ at $a$ and $c$, we have
$v_5(t_{q_1}) = v_4(t_{q_1}) = v_3^k(t_{q_1})$. From Phase 2, this
value equals $v(t_{q_1}) + k$. Secondly, notice that
$v_4(t_{r_1}) = v_3(t_{r_1})$, which from Phase 1 equals
$v(t_{p_1})$. From the equality $v_4(t_{q_1}) = v_4(t_{r_1})$, we get
$v(t_{q_1}) + k = v(t_{p_1})$.  This shows that
$k = v(t_{p_1} - v_{t_{q_1}})$. The number of times the loop $b$ is
done equals the difference between $t_{p_1}$ and $t_{q_1}$ at the
start of this gadget. If the number of $b$ iterations is more than
this value of less than this value, the
the negotiation cannot proceed further.

When action $c$ is done, the clock $x'_{q_1}$ holds the time between
$st$ and $c$ for agent $q_1$, and this is exactly $k$. If $k > 0$ the
value of $t_{q_1}$ is increased by $1$, resulting in the counter value
getting decremented and the agents move to the next instruction (shown
as the decrement gadget in the figure). Otherwise, the agents are sent
to a gadget from which there is no run (shown as the block gadget in
the figure).
Notice the interplay between synchronized and unsynchronized nodes in
this gadget. It is crucial that $n_2$ is unsynchronized, whereas
nodes $n_1, n_3$ need to be.

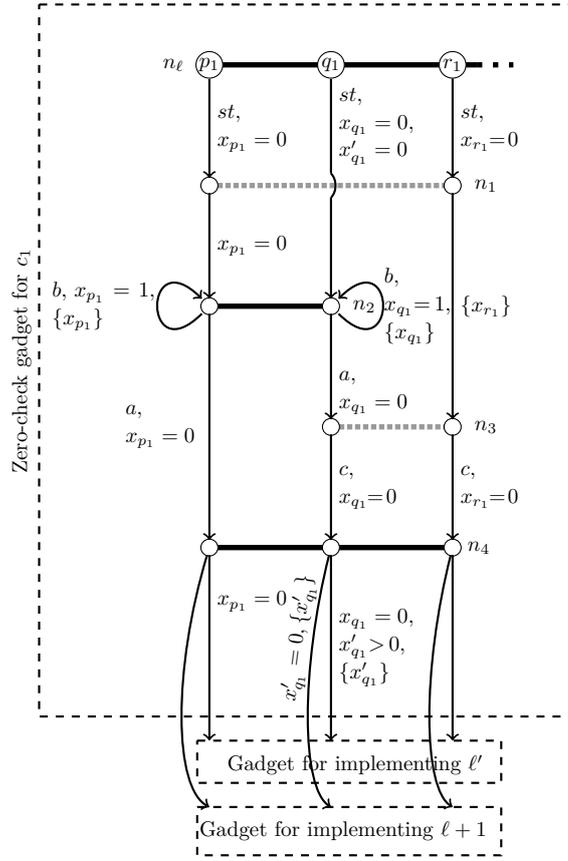
\begin{figure}[t]
\centering
\begin{tikzpicture}[scale=0.8, every node/.style={transform shape}]

\begin{scope}[shift={(10,0)}]

\begin{scope}[line width=2pt]
\draw (0,0) -- (2.2*\dx,0);
\draw [loosely dotted] (2.2*\dx,0) -- (2.5*\dx,0);
\draw [densely dotted, opacity=0.4] (0,-1*\dx) -- (2*\dx+0,-1*\dx);
\draw (0,-2*\dx) -- (\dx+0,-2*\dx);
\draw [densely dotted, opacity=0.4] (\dx-0,-3*\dx) -- (2*\dx+0,-3*\dx);
\draw (0,-4*\dx) -- (2*\dx,-4*\dx);
\end{scope}

\filldraw [fill=white]  (0,0)    circle (2.2mm) node (nip1) [label={left:$n_{\ell}$}] {$p_1$};
\filldraw [fill=white]  (\dx,0)  circle (2.2mm) node (niq1) {$q_1$};
\filldraw [fill=white]  (2*\dx,0) circle (2.2mm) node (nir1) {$r_1$};

\filldraw [fill=white] (0,-\dx)			circle	 (1.4mm)	 node	 (z1p1)	{};
\filldraw [fill=white] (2*\dx,-\dx)		circle	 (1.4mm)	 node	 (z1r1) [label={right:$\,\,\textcolor{black}{n_1}$}]	{};

\filldraw [fill=white] (0,-2*\dx)		circle	 (1.4mm)	 node	 (z2p1)	{};
\filldraw [fill=white] (\dx,-2*\dx)		circle	 (1.4mm)	 node	 (z2q1) [label={right:$\,\,n_2$}]	{};

\filldraw [fill=white] (\dx,-3*\dx)		circle	 (1.4mm)	 node	 (z3q1)	{};
\filldraw [fill=white] (2*\dx,-3*\dx)	circle	 (1.4mm)	 node	 (z3r1) [label={right:$\,\,\textcolor{black}{n_3}$}]	{};

\filldraw [fill=white] (0,-4*\dx)		circle	 (1.4mm)	 node	 (z4p1)	{};
\filldraw [fill=white] (\dx,-4*\dx)		circle	 (1.4mm)	 node	 (z4q1)	{};
\filldraw [fill=white] (2*\dx,-4*\dx)	circle	 (1.4mm)	 node	 (z4r1) [label={right:$n_4$}]	{};

\begin{scope}[thick,->]  
  \draw (nip1) -- (z1p1) node [midway, right, text width=1.25cm] {$st,$\\ $x_{p_1} = 0$};
  \draw (z1p1) -- (z2p1) node [midway, right, text width=1.25cm] {$x_{p_1}=0$};
  \draw (z2p1) -- (z4p1) node [midway, left, text width=1.25cm] {$a,$ $x_{p_1} = 0$};
  \draw (z4p1) -- (0*\dx,-5.6*\dx)  node [near start, right, text width=1.25cm] {$x_{p_1}=0$};
  \draw (z2p1) .. controls (-1.1,-2*\dx-1.1) and (-1.1,-2*\dx+1.1) .. node [midway, left=-0.4cm, text width=2cm] {$b,$ $x_{p_1}=1,$ $\{x_{p_1}\}$} (z2p1);
  
  \draw (\dx,-\dx-.2) -- (z2q1) node [midway, right] {};
  \draw (z2q1) -- node [near end, right, text width=1.25cm] {$a,$\\ $x_{q_1} = 0$} (z3q1);
  \draw (z3q1) -- (z4q1) node [midway, right, text width=1.25cm] {$c,$\\ $x_{q_1}\!\!=\!0$};
  \draw (z4q1) --  node [near start, right] {} (1*\dx,-5.6*\dx) node [midway=0.3cm, right, text width=1.25cm] {$x_{q_1} = 0,$ $x_{q_1}'\!\!>\!0,$ $\{x_{q_1}'\}$};
  \draw (z2q1) .. controls (\dx+1.1,-2*\dx-1.1) and (\dx+1.1,-2*\dx+1.1) .. node [midway, right=-0.11cm, text width=1.25cm] {$b,$\\ $x_{q_1}\!\!=\!1,$ $\{x_{q_1}\}$} (z2q1);
  
  \draw (nir1) -- (z1r1) node [midway, right, text width=1.25cm] {$st,$\\ $x_{r_1}\!\!=\!0$};
  \draw (z1r1) -- (z3r1) node [midway, right, text width=1.25cm] {$\{x_{r_1}\}$};
  \draw (z3r1) -- (z4r1) node [midway, right, text width=1cm] {$c,$ $x_{r_1}\!\!=\!0$};
  \draw (z4r1) -- (2*\dx,-5.6*\dx) node [near start, right] {};

 \draw (z4p1) .. controls (-.6,-10) and (-.6,-12) .. node [very near start, sloped, above] {} (0*\dx,-6.15*\dx) {};

 \draw (z4q1) .. controls (1.5,-10) and (1.5,-12) .. node [near start, sloped, above=-0.1cm] {$x'_{q_1} = 0, \{x_{q_1}'\}$} (1*\dx,-6.15*\dx) {};

 \draw (z4r1) .. controls (3.5,-10) and (3.5,-12) .. node [very near start, sloped, above] {} (2*\dx,-6.15*\dx);

\end{scope}

\draw [thick] (niq1) -- (\dx,-\dx+.2) node [midway, right, text width=1.25cm] {$st,$ $x_{q_1} = 0,$ $x_{q_1}'=0$};
\draw [thick] (\dx,-\dx+.2) .. controls (\dx+.1,-\dx+.2-.1) and(\dx+.1,-\dx-.2+.1) .. (\dx,-\dx-.2);

\draw [dashed, thick] (-1.5*\dx+.2,1) -- (3*\dx,1) -- (3*\dx,-5.4*\dx) -- (-1.5*\dx+.2,-5.4*\dx) -- (-1.5*\dx+.2,1) node [midway, sloped, above] {Zero-check gadget for $c_1$};

\draw [dashed, thick] (-.1*\dx,-5.6*\dx) -- (2.4*\dx,-5.6*\dx) -- (2.4*\dx,-5.95*\dx) -- (-.1*\dx,-5.95*\dx) -- (-.1*\dx,-5.6*\dx);

\draw [dashed, thick] (-.1*\dx,-6.15*\dx) -- (2.4*\dx,-6.15*\dx) -- (2.4*\dx,-6.55*\dx) -- (-.1*\dx,-6.55*\dx) -- (-.1*\dx,-6.15*\dx);

\node at (1.2*\dx,-5.8*\dx) {Gadget for implementing $\ell'$};
\node at (1.1*\dx,-6.35*\dx) {Gadget for implementing $\ell +1$};
\end{scope}
\end{tikzpicture}
\caption{Gadget for implementing jump-on-zero instruction on $c_1$ }
\label{fig:jump}
\end{figure}

\paragraph*{Jump-on-zero.} This gadget is similar to the decrement
gadget, where the first part was to check if $C_1$ is $0$ or not. When $C_1 == 0$, the gadget jumps to the
relevant instruction, otherwise it moves to the next instruction in
sequence. The gadget is shown in Figure~\ref{fig:jump}.

\section{Conclusion}
\label{sec:conclusion}

We have presented a model of local-timed negotiations. This is
motivated by the need for expressing timing constraints between
interactions in a negotiation. We have chosen a local-time model and
incorporated a synchronization constraint as part of the model. We
have shown that reachability is decidable when there is no mix of
synchronized and unsynchronized interactions. This mix creates
situations where one agent needs to fire a number of outcomes before
synchronizing with a second agent, and in this process forces a third
agent to elapse time. We have used this in the gadget explained for
the decrement instruction. As future work, we would like to study
non-trivial restrictions which contain a mix of synchronized and
unsynchronized interactions and are yet decidable.

We would like to remark that such a synchronization constraint can be
added in the local-time semantics for networks of timed
automata. Currently, the local-time semantics forces every shared
action to be synchronized. The main reason is that the definition
gives equivalence with the global-time semantics for reachability and
B\"uchi reachability. For networks of timed automata, global-time
semantics is considered the gold standard. The local-time semantics is
studied as a heuristic to solve reachability and B\"uchi reachability,
since this has better independence properties and is therefore
amenable to partial-order methods.  In our case with negotiations, we
consider local-time as the original semantics and make synchronization
as an option to be specified in the model. This allows more
independence between the agents and makes it more attractive for
partial-order methods. Having a decidable fragment with controlled use
of synchronization would be interesting in this regard.

\bibliographystyle{plain}
\bibliography{timednegotiation}

\end{document}